\documentclass{myllncs}
\usepackage{graphicx}
\usepackage{amsmath}
\usepackage{amssymb}
\usepackage[ruled, lined, algo2e, linesnumbered]{algorithm2e}
\usepackage[T1]{fontenc}
\usepackage{enumerate}
\usepackage[font=small, format=hang, labelfont=bf]{caption}
\usepackage[]{caption}
\usepackage{subfig}
\usepackage{color}

\definecolor{orange}{rgb}{1,0.5,0}
\newcommand{\cmt}[1]{$\{$#1$\}$}
\newcommand{\comment}[1]{//#1}
\SetArgSty{}  \SetKw{KwOr}{or}
\SetKw{KwAnd}{and} \SetKw{KwInvariant}{invariant:}
\SetKw{KwDownTo}{down to} \SetKwInOut{Input}{Input}
\SetKwInOut{Output}{Output} \SetKwInOut{Require}{require}
\SetKwComment{Comment}{start}{end}
\SetKwIF{If}{ElseIf}{Else}{if}{then}{else if}{else}{endif}

\newcommand{\NP}{\mathcal{NP}}
\newcommand{\spring}{\mathtt{spring}}

\newcommand{\cros}{\mathtt{cros}}
\newcommand{\ang}{\mathtt{angle}}
\newcommand{\angular}{\mathtt{angular}}

\newcommand{\eat}[1] {{}}
\newcommand{\nplus}[1]{$=^{^{\hspace{-8pt}\small(#1)}}$}
\definecolor{orange}{rgb}{1,0.5,0}
\begin{document}

\title{Maximizing the Total Resolution of Graphs}

\titlerunning{Maximizing the Total Resolution of Graphs}

\toctitle{Maximizing the Total Resolution of Graphs}

\author{Evmorfia N.\ Argyriou \inst{1}, Michael A.\ Bekos \inst{1} and Antonios Symvonis \inst{1}}

\authorrunning{E.N. Argyriou, M.A. Bekos, A. Symvonis}

\tocauthor{Evmorfia N. Argyriou, Michael A. Bekos, Antonios Symvonis}

\authorrunning{E.N. Argyriou, M.A. Bekos, A. Symvonis}

\institute{%
    School of Applied Mathematical \& Physical Sciences,\\
    National Technical University of Athens, Greece\\
    \email{$\{$fargyriou,mikebekos,symvonis$\}$@math.ntua.gr}
}

\maketitle              

\begin{abstract}

A major factor affecting the readability of a graph drawing is its
resolution. In the graph drawing literature, the resolution of a
drawing is either measured based on the angles formed by consecutive
edges incident to a common node (angular resolution) or by the
angles formed at edge crossings (crossing resolution). In this
paper, we evaluate both by introducing the notion of ``total
resolution'', that is, the minimum of the angular and crossing
resolution. To the best of our knowledge, this is the first time
where the problem of maximizing the total resolution of a drawing is
studied.

The main contribution of the paper consists of drawings of
asymptotically optimal total resolution for complete graphs
(circular drawings) and for complete bipartite graphs ($2$-layered
drawings). In addition, we present and experimentally evaluate a
force-directed based algorithm that constructs drawings of large
total resolution.\\~\\Date: \emph{September 10, 2010.}

\end{abstract}

\section{Introduction}
\label{sec:introduction}

Graphs are widely used to depict relations between objects. There
exist several criteria that have been used to judge the quality of a
graph drawing \cite{DBTT94,KW01}. From a human point of view, it is
necessary to obtain drawings that are easy-to-read, i.e., they
should nicely convey the structure of the objects and their
relationships. From an algorithmic point of view, the quality of a
drawing is usually evaluated by some objective function and the main
task is to determine a drawing that minimizes or maximizes the
specific objective function. Various such functions have been
studied by the graph drawing community, among them, the number of
crossings among pairs of edges, the number of edge bends, the
maximum edge length, the total area occupied by the drawing and so
on.

Over the last few decades, much research effort has been devoted to
the problem of reducing the number of crossings. This is reasonable,
since it is commonly accepted that edge crossings may negatively
affect the quality of a drawing. Towards this direction, there also
exist eye-tracking experiments that confirm the negative impact of
edge crossings on the human understanding of a graph drawing
\cite{P00,PCA02,CPCM02}. However, the computational complexity of
the edge crossing minimization problem, which is $\NP$-complete in
general \cite{GJ83}, implies that the computation of high-quality
drawings of dense graph is difficult to achieve.

Apart from the edge crossings, another undesired property that may
negatively influence the readability of a drawing is the presence of
edges that are too close to each other, especially if these edges
are adjacent. Thus, maximizing the angles among incident edges
becomes an important aesthetic criterion, since there is some
correlation between the involved angles and the visual
distinctiveness of the edges.

Motivated by the cognitive experiments by Huang et al.\
\cite{Hu07,HHE08} that indicate that the negative impact of an edge
crossing is eliminated in the case where the crossing angle is
greater than $70$ degrees, we study a new graph drawing scenario in
which both \emph{angular} and \emph{crossing
resolution}\footnote{The term \emph{angular resolution} denotes the
smallest angle formed by two adjacent edges incident to a common
node, whereas the term \emph{crossing resolution} refers to the
smallest angle formed by a pair of crossing edges.} are taken into
account in order to produce a straight-line drawing of a given
graph. To the best of our knowledge, this is the first attempt,
where both angular and crossing resolution are combined to produce
drawings.
We prove that the classes of complete and complete bipartite graphs
admit drawings that asymptotically maximize the minimum of the
angular and crossing resolution (Section~\ref{sec:complete-graphs}).
We also present a more practical, force-directed based algorithm
that constructs drawings of large angular and crossing resolution
(Section~\ref{sec:force-directed-algorithm}).

\subsection{Previous Work}
\label{sec:previous-work}

Formann et al.\ \cite{FHHKLSWW93} were the first to study the
angular resolution of straight-line drawings. They proved that
deciding whether a graph of maximum degree $d$ admits a drawing of
angular resolution $\frac{2 \pi}{d}$ (i.e., the obvious upper bound) is
$\NP$-hard. They also proved that several types of graphs of
maximum degree $d$
have angular resolution $\Theta(\frac{1}{d})$. Malitz and Papakostas
\cite{MP92} proved that any planar graph of maximum degree $d$,
admits a planar straight-line drawing with angular resolution
$\Omega(\frac{1}{7^d})$. Garg and Tamassia \cite{GT94} showed a
continuous tradeoff between the area and the angular resolution of
planar straight-line drawings. Gutwenger and Mutzel \cite{GM98} gave
a linear time and space algorithm that constructs a planar polyline
grid drawing of a connected planar graph with $n$ vertices and
maximum degree $d$ on a $(2n-5)\times(\frac{3}{2}n-\frac{7}{2})$
grid with at most $5n-15$ bends and minimum angle greater than
$\frac{2}{d}$. Bodlaender and Tel \cite{BT04} showed that planar
graphs with angular resolution at least $\frac{\pi}{2}$ are
rectilinear. Recently, Lin and Yen \cite{LY05} presented a
force-directed method based on edge-edge repulsion that leads to
drawings with high angular resolution. In their work, pairs of edges
incident to a common node are modeled as charged springs, that repel
each other.


A graph is called \emph{right angle crossing} (or \emph{RAC} for
short) graph if it admits a polyline drawing in which every pair of
crossing edges intersects at right angle. Didimo et al.\
\cite{DEL09}
showed that any straight-line RAC drawing with $n$ nodes has at most
$4n-10$ edges. Angelini et al.\ \cite{ACBDFKS09} showed that there
are acyclic planar digraphs not admitting straight-line upward
RAC drawings and that the corresponding decision problem is
$\NP$-hard. They also constructed digraphs whose straight-line
upward RAC drawings require exponential area. Di Giacomo et al.\
\cite{DGDLM10} presented tradeoffs between the crossing resolution,
the maximum number of bends per edges and the area. Dujmovic  et
al.\ \cite{DGMW10} studied \emph{$\alpha$~Angle Crossing} (or
\emph{$\alpha$AC} for short) graphs, that generalize the RAC graphs.
A graph is called $\alpha$AC if it admits a polyline drawing
in which the smallest angle formed by an edge crossing is at least $\alpha$. For this class of
graphs, they proved upper and lower bounds for the number of edges.

Force-directed methods are commonly used for drawing graphs
\cite{Ea84,FR91}. In such a framework, a graph is treated as a
physical system with forces acting on it. Then, a good configuration
or drawing can be obtained from an equilibrium state of the system.
An overview of force-directed methods and their variations can be
found in the graph drawing books \cite{DBTT94,KW01}.


\section{Preliminaries and Notation}
\label{sec:preliminaries}

Let $G=(V,E)$ be an undirected graph. Given a drawing $\Gamma(G)$ of
$G$, we denote by $p_u=(x_u,y_u)$ the position of node $u \in V$ on
the plane. The unit length vector from $p_u$ to $p_v$ is denoted, by
$\overrightarrow{p_up_v}$, where $u,v\in V$. The degree of node
$u\in V$ is denoted by $d(u)$. Let also $d(G)=\max_{u\in V}d(u)$ be
the degree of the graph.

Given a pair of points $q_1, q_2 \in \mathbb{R}^2$, with a slight
abuse of notation, we denote by $||q_1-q_2||$ the Euclidean distance
between $q_1$ and $q_2$. We refer to the line segment defined by
$q_1$ and $q_2$ as $\overline{q_1q_2}$.

Let $\overrightarrow{\alpha}$ and $\overrightarrow{\gamma}$ be two
vectors. The vector which bisects the angle between
$\overrightarrow{\alpha}$ and $\overrightarrow{\gamma}$ is
$\frac{\overrightarrow{\alpha}}{||\overrightarrow{\alpha}||}+\frac{\overrightarrow{\gamma}}{||\overrightarrow{\gamma}||}$.
We denote by
$\mathtt{Bsc}(\overrightarrow{\alpha},\overrightarrow{\gamma})$ the
corresponding unit length vector. Given a vector
$\overrightarrow{\beta}$, we refer to the unit length vector which
is perpendicular to $\overrightarrow{\beta}$ and precedes it in the
clockwise direction, as $\mathtt{Perp}(\overrightarrow{\beta})$.
Some of our proofs use the following elementary geometric
properties:


\begin{center}
\begin{minipage}[b]{.50\textwidth}
\centering
\begin{equation}
\label{eq:tan-minus} \tan{(\omega_1-\omega_2)}
=\frac{\tan{\omega_1}- \tan{\omega_2}}{1+\tan{\omega_1} \cdot
\tan{\omega_2}}
\end{equation}
\end{minipage}
\begin{minipage}[b]{.48\textwidth}
\centering
\begin{equation}
\label{eq:tan-divide} \tan{(\omega/2)}
=\frac{\sin{\omega}}{1+\cos{\omega}}
\end{equation}
\end{minipage}
\begin{minipage}[b]{\textwidth}
~
\end{minipage}
\begin{minipage}[b]{.50\textwidth}
\begin{equation}
\label{eq:tan-inequality} \omega\in(0,\frac{\pi}{2}) \Rightarrow \tan{\omega} >\omega
\end{equation}
\end{minipage}
\begin{minipage}[b]{.48\textwidth}
~
\end{minipage}
\end{center}

\section{Drawings with Optimal Total Resolution for Complete and Complete Bipartite Graphs}
\label{sec:complete-graphs}

In this section, we define the total resolution of a drawing and we
present drawings of asymptotically optimal total resolution for
complete graphs (circular drawings) and complete bipartite graphs
($2$-layered drawings).

\begin{definition}
The \emph{total resolution} of a drawing is defined as the minimum
of its angular and crossing resolution.
\end{definition}
%

We first consider the case of complete graphs. Let $K_{n}=(V,E)$ be
a complete graph, where $V=\{u_0, u_1, \ldots, u_{n-1}\}$ and
$E=V\times V$. Our aim is to construct a circular drawing of $K_n$
of maximum total resolution. Our approach is constructive and common
when dealing with complete graphs. A similar one has been given by
Formann et al. \ \cite{FHHKLSWW93} for obtaining optimal drawings of
complete graphs, in terms of angular resolution. Consider a circle
$\mathcal{C}$ of radius $r_c>0$ centered at $(0,0)$ and circumscribe
a regular $n$-polygon $\mathcal{Q}$ on $\mathcal{C}$. In our
construction, the nodes of $K_n$ coincide with the vertices of
$\mathcal{Q}$. W.l.o.g., we further assume that $u_1, u_2, \ldots,
u_n$ appear in this order in the counter-clockwise direction around
$(0,0)$, as illustrated in Fig.\ref{fig:complete-case}.

\begin{theorem}
A complete graph $K_n$ admits a drawing of
total resolution $\Theta (\frac{1}{n})$.
\label{thm:complete-total-resolution}
\end{theorem}

\begin{proof}
We prove that the angular resolution of the presented drawing of
$K_n$ is $\frac{\pi}{n}$, whereas its crossing resolution is
$\frac{2\pi}{n}$. First, observe that the arc of circle
$\mathcal{C}$ that connects two consecutive nodes $u_i$ and
$u_{(i+1)modn}$ is equal to $\frac{2\pi}{n}$, for each
$i=0,1,\dots,n-1$. Therefore, the angular resolution of the drawing
is $\frac{\pi}{n}$, as desired. Let now $e_i=(u_i,u_{i'})$ and
$e_j=(u_j,u_{j'})$ be two crossing edges. Without loss of
generality, we assume that $i<j<i'<j'$, as in
Fig.\ref{fig:complete-case}. The crossing of $e_i$ and $e_j$ defines
two angles $\phi_c$ and $\phi_c'$ such that $\phi_c+\phi_c'=\pi$. In
Fig.\ref{fig:complete-case}, $\phi_c$ is exterior to the triangle
formed by the crossing of $e_i$ and $e_j$ and the nodes $u_j$ and
$u_{i'}$ (refer to the dark-gray triangle of
Fig.\ref{fig:complete-case}). Therefore: $\phi_c =
(j'-i')\frac{\pi}{n} + (j-i)\frac{\pi}{n}$. Similarly,
$\phi_c'=(i'-j)\frac{\pi}{n}+(n-(j'-i))\frac{\pi}{n}$. In the case,
where $j=(i+1)\mod{n}$ and $j'=(i'+1)\mod{n}$ (i.e., the nodes $u_i$
($u_{i'}$, resp.) and $u_j$ ($u_{j'}$, resp.) are consecutive), the
angle $\phi_c$ receives its minimum value, which is equal to
$\frac{2\pi}{n}$. Similarly, we can prove that the minimum value of
$\phi_c'$ is also $\frac{2\pi}{n}$. This establishes that the
crossing resolution is $\frac{2\pi}{n}$.\qed
\end{proof}



We now proceed to consider the class of complete bipartite graphs.
Since an $n$-vertex complete bipartite graph is a subgraph of a
$n$-vertex complete graph, the bound of the total resolution of a
complete bipartite graph can be implied by the bound of the complete
graph. However, if the nodes of the graph must have integer
coordinates, i.e., we restrict ourselves on grid drawings, few
results are known regarding the area needed of such a drawing. An
upper bound of $O(n^3)$ area can be implied by \cite{BT04}. This
motivates us to separately study the class of complete bipartite
graph, since we can drastically improve this bound. Note that the
tradeoff between resolution and area has been studied by the graph
drawing community, in the past. Malitz and Papakostas \cite{MP92}
showed there exist graphs that always require exponential area for
straight-line embeddings maintaining good angular resolution. The
claim remains true, if circular arc edges are used instead of
straight-line \cite{CDGK99}. More recently, Angelini et al.\
\cite{ACBDFKS09} constructively showed that there exists graphs
whose straight-line upward RAC drawings require exponential area.


Again, we follow a constructive approach. First, we consider a
square $\mathcal{R}=AB\Gamma\Delta$ where its top and bottom sides
coincide with $\mathcal{L}_1$ and $\mathcal{L}_2$, respectively (see
Fig.\ref{fig:bipartite-case}). Let $H$ be the height (and width) of
$\mathcal{R}$. According to our approach, the nodes of $V_1$ ($V_2$,
resp.) reside along side $\Gamma \Delta$ ($AB$, resp.) of
$\mathcal{R}$. In order to specify the exact positions of the nodes
$u^1_1, u^1_2, \ldots, u^1_m$ along side $\Gamma\Delta$, we first
construct a bundle of $m$ semi-lines, say $\ell_1, \ldots, \ell_m$,
each of which emanates from vertex $B$ and crosses side
$\Gamma\Delta$ of $\mathcal{R}$, so that the angle formed by
$B\Gamma$ and semi-line $\ell_i$ equals to $\frac{(i-1)\cdot
\widehat{\Delta B \Gamma}}{m-1}$, for each $i=1, \ldots, m$. These
semi-lines split angle $\widehat{\Delta B \Gamma}$ into $m-1$
angles, each of which is equal to $\frac{\pi}{4 \cdot (m-1)}$, since
$\widehat{\Delta B \Gamma}=\pi/4$. Say $\phi =\frac{\pi}{4 \cdot
(m-1)}$. Then, we place node $u^1_i$ at the intersection of
semi-line $l_i$ and $\Gamma\Delta$, for each $i=1, \ldots, m$ (see
Fig.\ref{fig:bipartite-case}). In order to simplify the description
of our approach, we denote by $a_i$ the horizontal distance between
two consecutive nodes $u^1_i$ and $u^1_{i+1}$, $i = 1, \ldots, m-1$.

We proceed by defining an additional bundle of $m$ semi-lines, say
$\ell_1', \ldots, \ell_m'$, that emanate from vertex $A$. More
precisely, semi-line $l_{i}'$ emanates from vertex $A$ and passes
through the intersection of $l_{m-i}$ and $\Gamma\Delta$ (i.e., node
$u^1_{m-i}$), for each $i=1, \ldots, m$ (see
Fig.\ref{fig:bipartite-case}). Let $\phi_i'$ be the angle formed by
two consecutive semi-lines $l_i'$ and $l_{i+1}'$, for each
$i=1,\ldots,m-1$.

\begin{figure}[t!hb]
  \centering
  \begin{minipage}[b]{.4\textwidth}
    \centering
    \subfloat[\label{fig:complete-case}{A circular drawing of $K_n$.}]
    {\includegraphics[width=.90\textwidth]{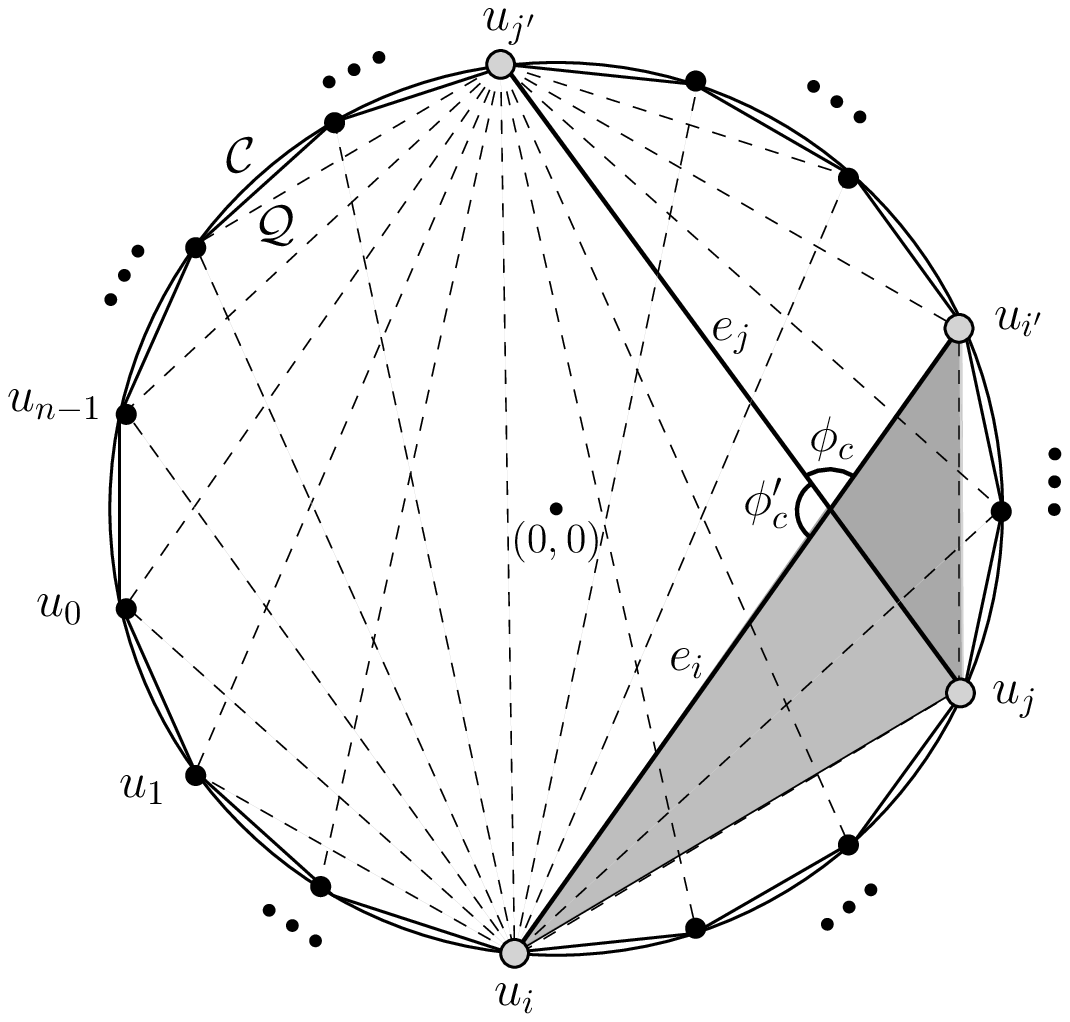}}
  \end{minipage}
  \begin{minipage}[b]{.57\textwidth}
    \centering
    \subfloat[\label{fig:bipartite-case}{A $2$-layered drawing of $K_{m,n}$.}]
    {\includegraphics[width=.90\textwidth]{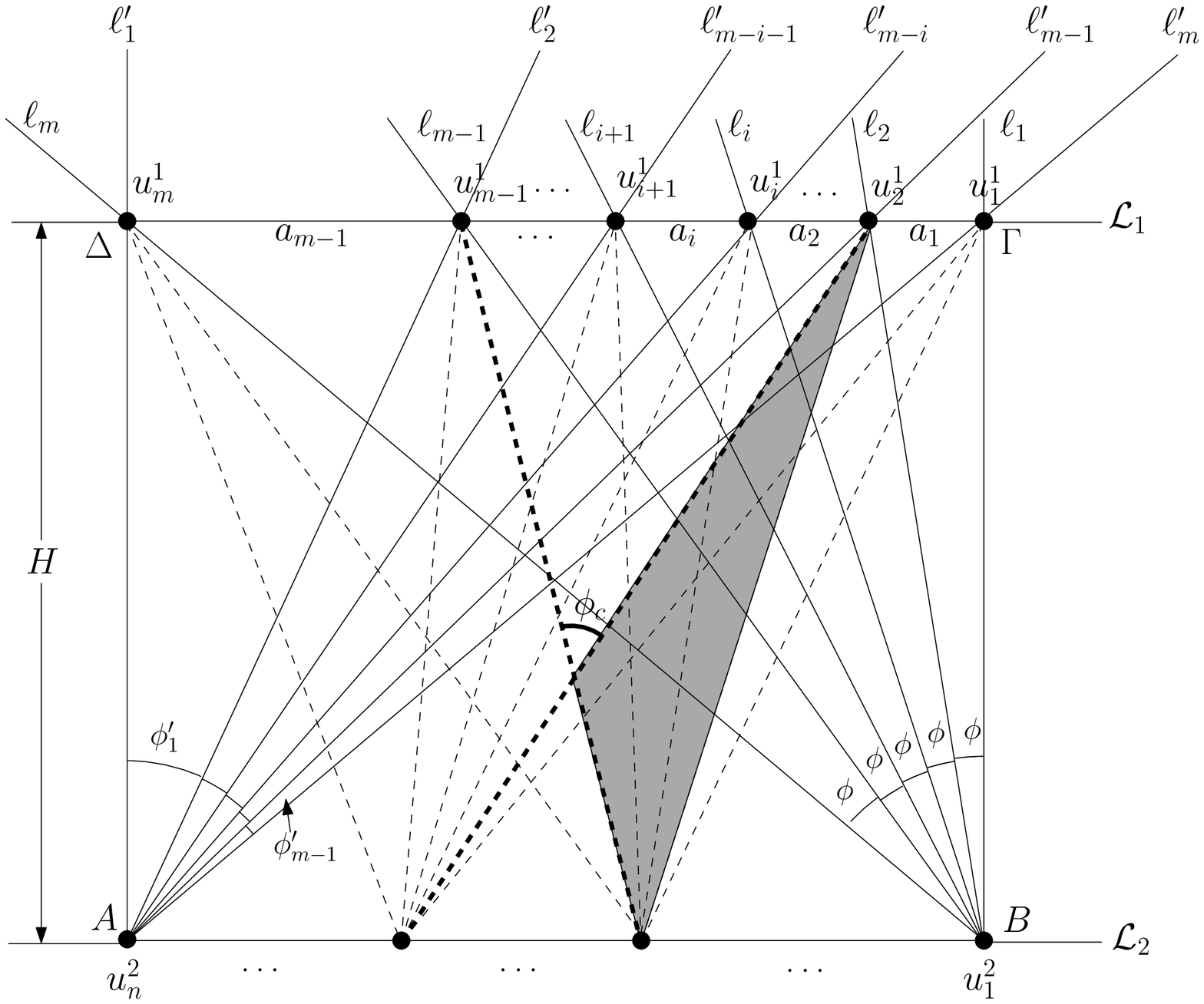}}
  \end{minipage}
  \caption{Illustrations of our constructions}
  \label{fig:construction-ilustrations}
\end{figure}

So far, we have managed to fix the position of the nodes of $V_1$
only (along side $\Gamma\Delta$ of $\mathcal{R}$). Symmetrically, we
define the position of the nodes of $V_2$ along side $AB$ of
$\mathcal{R}$. This only involves two additional bundles of
semi-lines emanating from vertices $\Gamma$ and $\Delta$. We now
proceed to investigate some geometric properties of the proposed
construction.

\begin{lemma}
For each $i=1,2,\ldots, m-1$, it holds that $a_{i-1} < a_i$.
\label{lem:line-segment-comparison}
\end{lemma}

\begin{proof}
By induction. For the base of the induction, we have to show that
$a_1 < a_2$. First observe that $a_1= H\tan{\phi}$ and $a_1 + a_2 =
H \tan{2\phi}$. Therefore:

\begin{center}
$a_2 = H (\tan{2\phi} - \tan{\phi}) $ \nplus{\ref{eq:tan-minus}} $
a_1 \cdot (1+\tan{2\phi} \cdot \tan{\phi})$
\end{center}

However, both $\tan{\phi}$ and $\tan{2\phi}$ are greater than zero,
which immediately implies that $a_1<a_2$. For the induction
hypothesis, we assume that $\forall$ $k$, $k<m-1$ it holds that
$a_{k-1} < a_k$ and we should prove that $a_k < a_{k+1}$. Obviously,
$a_1 + \ldots + a_k = H \tan{k\phi}$. Based on
Equation~\ref{eq:tan-minus} and similarly to the base of the
induction, we have:

\begin{itemize}
  \item $a_{k+1} = H \tan{\phi} \cdot(1+\tan{(k+1)\phi} \cdot \tan{k \phi})$
  \item $a_k ~~~= H \tan{\phi} \cdot (1+\tan{(k-1)\phi} \cdot \tan{k\phi} )$
\end{itemize}

In order to complete the proof, observe that $(k+1)\phi >
(k-1)\phi$.\qed
\end{proof}

\begin{lemma}
For each $i=1,2,\ldots, m-1$, it holds that $\phi_{i-1}' > \phi_i'$.
\label{lem:angle-comparison}
\end{lemma}

\begin{proof}
By induction. For the base of the induction, we have to prove that
$\phi_1' > \phi_2'$ or equivalently that $\tan{\phi_1'} >
\tan{\phi_2'}$. It holds that $\tan{\phi_1'}=a_{m-1}/H$ and
$\tan{(\phi_1' + \phi_2')} = (a_{m-1} + a_{m-2})/H$. By combining
these relationships with Equation~\ref{eq:tan-minus} we have that
$\tan{\phi_2'} = \frac{Ha_{m-1}}{H^2+a^2_{m-1} + a_{m-1}a_{m-2}}$.
Therefore:

$$\tan\phi_1' > \tan\phi_2' \Leftrightarrow H^2(a_{m-1}-a_{m-2})+a^3_{m-1} + a^2_{m-1}a_{m-2}
> 0,$$

\noindent which trivially holds due to
Lemma~\ref{lem:line-segment-comparison}. For the induction
hypothesis, we assume that $\forall$ $k$, $k<m-1$ it holds that
$\phi_{k-1}' > \phi_{k}'$ and we have to show that $\phi_{k}'
> \phi_{k+1}'$. Observe that:

\begin{itemize}
  \item $\tan\phi_{k}' = \frac{\tan(\phi_1+ \ldots + \phi_{k}') -
\tan(\phi_1'+ \ldots + \phi_{k-1}')}{1+\tan(\phi_1'+ \ldots +
\phi_{k}') \cdot \tan(\phi_1'+ \ldots + \phi_{k-1}')}$ = $\frac{H
a_{m-k}}{H^2+(a_{m-1} + \ldots + a_{m-k})(a_{m-1} + \ldots +
a_{m-k+1})}$
  \item $\tan(\phi_{k+1}')
= \dots = \frac{H a_{m-(k+1)}}{H^2+(a_{m-1} + \ldots +
a_{m-(k+1)})(a_{m-1} + \ldots + a_{m-k})}$
\end{itemize}

\noindent By Lemma~\ref{lem:line-segment-comparison} we have that $
(a_{m-1} + \ldots + a_{m-(k-1)}) > (a_{m-1} + \ldots + a_{m-(k+1)})$
and $ H \cdot a_{m-k} > H \cdot a_{m-(k+1)}$. Therefore,
$\tan\phi_{k}' > \tan\phi_{k+1}'$.\qed
\end{proof}

\begin{lemma}
Angle $\phi_{m-1}'$ is the smallest angle among all the angles
formed in the drawing. \label{lem:smallest-angle}
\end{lemma}

\begin{proof}
From Lemma \ref{lem:angle-comparison}, it follows that angle
$\phi_{m-1}'$ is the smallest angle among all $\phi_i'$, $i=1,
\ldots, m-1$. Additionally, it is not difficult to see that angle
$\phi_{m-1}'$ is larger (but remains the smallest among all
$\phi_i'$, $i=1, \ldots, m-1$), if the endpoint of the bundle (i.e.,
node $u^2_n$) moves to any internal point (i.e., node $u^2_i$) of
side $AB$ (see Fig.\ref{fig:bipartite-case}). Therefore,
$\phi_{m-1}'$ is the smallest angle among all the angles formed by
pairs of consecutive edges incident to any node of $V_2$. Since $m
\geq n$, the same holds for the nodes of $V_1$. Therefore,
$\phi_{m-1}'$ defines the angular resolution of the drawing.

Consider now two crossing edges (refer to the bold, crossing
dashed-edges of Fig.\ref{fig:bipartite-case}). Their crossing
defines (a)~a pair of angles that are smaller than $90^o$ and
(b)~another pair of angles that are larger than $90^o$. Obviously,
only the acute angles participate in the computation of the crossing
resolution (see angle $\phi_c$ in Fig.\ref{fig:bipartite-case}).
However, in a complete bipartite graph the acute angles are always
exterior to a triangle having two of its vertices on $V_1$ and
$V_2$, respectively (refer to the gray-colored triangle of
Fig.\ref{fig:bipartite-case}). Therefore, the crossing resolution is
always greater than the angular resolution, as desired. \qed

\end{proof}

\begin{lemma}
It holds that $\phi_{m-1}' \geq \frac{\phi}{2}$.
\label{lem:smallest-angle-phi}
\end{lemma}

\begin{proof}
We equivalently prove that $\tan\phi_{m-1}' > \tan\frac{\phi}{2} $.
Using Equation~\ref{eq:tan-divide}, we have that $\tan\frac{\phi}{2}
< \frac{a_1}{2 H}$. Therefore:

\begin{align*}
\tan\phi_{m-1}' > \tan\frac{\phi}{2} &\Longleftrightarrow
\frac{\frac{a_1}{H}}{1+ \frac{a_1 + \ldots + a_{m-1}}{H} \cdot
\frac{a_2 + \ldots + a_{m-1}}{H} } > \frac{a_1}{2 \cdot H}\\
&\Longleftrightarrow H^2 > (a_1 + \ldots + a_{m-1}) (a_2 + \ldots +
a_{m-1})\\ & \Longleftrightarrow H \cdot a_1 > 0
\end{align*}

\noindent which obviously holds.\qed

\end{proof}

\begin{theorem}
A complete bipartite graph $K_{m,n}$ admits a $2$-layered drawing of
total resolution $\Theta (\frac{1}{\max\{{m,n}\}})$.
\label{thm:total-resolution}
\end{theorem}

\begin{proof}
Immediately follows from Lemmata~\ref{lem:smallest-angle} and
~\ref{lem:smallest-angle-phi}.\qed
\end{proof}



Consider now the case where the nodes of the graph must have integer
coordinates, i.e., we restrict ourselves on grid drawings. An
interesting problem that arises in this case is the estimation of
the total area occupied by the produced drawing. We will describe
how we can modify the positions of the nodes produced by our
algorithm in order to obey the grid constraints. Assume without loss
of generality that $\mathcal{L}_1$ and $\mathcal{L}_2$ are two
horizontal lines, so that $\mathcal{L}_2$ coincides with $y$-axis
and the drawing produced by our algorithm has $a_1=1$. Then, we can
express the height of drawing $\Gamma(K_{m,n})$ as a function of
$\phi$, as follows:

$$a_1=1 \Longleftrightarrow \tan{\phi} \cdot H =1 \Longleftrightarrow H = 1/\tan{\phi}$$

Note that this drawing does not obey the grid constraints. To
achieve this, we move the horizontal line $\mathcal{L}_1$ to the
horizontal grid line immediately above it and each node of both
$V_1$ and $V_2$ to the rightmost grid-point to its left. In this
manner, we obtain a new drawing $\Gamma'(K_{m,n})$, which is grid as
desired. By Lemma~\ref{lem:line-segment-comparison}, it follows that
there are no two nodes sharing the same grid point, since $a_1$ is
slightly greater than one grid unit. Since neither horizontal line
$\mathcal{L}_1$ nor any node of $K_{m,n}$ moves more than one unit
of length, the total resolution of $\Gamma'(K_{m,n})$ is not
asymptotically affected, and, in addition the height of the drawing
is not significantly greater (i.e., asymptotically it remains the
same). Based on the above, the area is bounded by $\cot^2{\phi}$ or
equivalently by $1/\tan^2{\phi}$. By
Equation~\ref{eq:tan-inequality}, this is further bounded by
$1/\phi^2$. By Theorem~\ref{thm:total-resolution}, it holds that
$\phi = O(1/\max\{m,n\})$. Therefore, the total area occupied by the
drawing is $O(\max\{m^2,n^2\})$. The following theorem summarizes
this result.


\begin{theorem}
A complete bipartite graph $K_{m,n}$ admits a $2$-layered grid
drawing of $\Theta (\frac{1}{\max\{{m,n}\}})$ total resolution and
$O(\max\{m^2,n^2\})$ area.
\end{theorem}

\section{A Force Directed Algorithm}
\label{sec:force-directed-algorithm}


We present a force-directed algorithm that given a reasonably nice
initial drawing, probably produced by a classical force-directed
algorithm, results in a drawing of high total resolution. The
algorithm reinforces the classical force-directed algorithm of Eades
\cite{Ea84} with some additional forces exerted to the nodes of the
graph. More precisely, these additional forces involve springs and
some extra attractive or repulsive forces on nodes with degree
greater than one and on end-nodes of edges that are involved in an
edge crossing. This aims to ensure that the angles between incident
edges and the angles formed by pairs of crossing edges will be as
large as possible.
The classical force-directed algorithm of Eades \cite{Ea84} models
the nodes of the graph as electrically charged particles that repel
each other, and its edges by springs in order to attract adjacent
nodes. In our approach, we use only the attractive forces of the
force-directed algorithm of Eades (denoted by
$\mathcal{F}_{\spring}$), which follow the formula:


$$\mathcal{F}_{\spring}(p_u, p_v) = C_{\spring} \cdot \log
{\frac{||p_u-p_v||}{\ell_{\spring}}}\cdot
\overrightarrow{p_up_v},~(u,v)\in E$$


\noindent where $C_{\spring}$ and $\ell_{\spring}$ capture the
stiffness and the natural length of the springs, respectively.
Recall that
$\overrightarrow{p_up_v}$ denotes the unit length vector from $p_u$
to $p_v$.

We first describe our approach for the case where two edges, say
$e=(u,v)$ and $e'=(u',v')$, are involved in a crossing. Let $p_c$ be
their intersection point.
W.l.o.g., we assume that $u$ and $u'$ are to the left of $v$ and
$v'$, respectively, $y_{u'}<y_{u}$ and $y_{v}<y_{v'}$, as in
Fig.\ref{fig:crossing case}. Let $\theta_{vv'}$ be the angle formed
by the line segments $\overline{p_cp_v}$ and $\overline{p_cp_{v'}}$
in counter-clockwise order around $u$ from $\overline{p_cp_v}$ to
$\overline{p_cp_{v'}}$. In order to avoid confusion, we assume that
$\theta_{vv'}=\theta_{v'v}$, i.e., we abuse the counter-clockwise
measurement of the angles that would result in
$\theta_{vv'}=2\pi-\theta_{v'v}$. Similarly, we define the remaining
angles of Fig.\ref{fig:crossing case}. Obviously,
$\theta_{vv'}+\theta_{v'u}=\pi$. Ideally, we would like
$\theta_{vv'}=\theta_{v'u}=\frac{\pi}{2}$, i.e., $e$ and $e'$ form a
right angle crossing. As we will shortly see, the magnitude of the
forces that we apply on the nodes $u,u',v$ and $v'$ depends on
(a)~the angles $\theta_{vv'}$ and $\theta_{v'u}$ and (b)~the lengths
of the line segments $\overline{p_cp_u}$, $\overline{p_cp_{u'}}$,
$\overline{p_cp_v}$ and $\overline{p_cp_{v'}}$.

The physical model that describes our approach is illustrated in
Fig.\ref{fig:crossing case}. Initially, for each pair of crossing
edges at point $p_c$, we place springs connecting consecutive nodes
in the counter-clockwise order around $p_c$, as in
Fig.\ref{fig:crossing-case-springs}. The magnitude of the forces due
to these springs should capture our preference for right angles.
Consider the spring connecting $v$ and $v'$. The remaining ones are
treated
symmetrically. We set the natural length, say
$\ell_{\spring}^{vv'}$, of the spring connecting the nodes $v$ and
$v'$ to be $\sqrt{||p_c-p_v||^2+||p_c-p_v'||^2}$. This quantity
corresponds to the length of the line segment that connects $v$ and
$v'$ in the optimal case where $\theta_{vv'}=\frac{\pi}{2}$. So, in
an equilibrium state of this model on a graph consisting only of $e$
and $e'$,
$e$ and $e'$ will form a right angle.
Concluding, the force on $v$ due to the spring of $v'$ is defined as
follows:

$$\mathcal{F}_{\spring}^{\cros}(p_v, p_{v'}) = C_{\spring}^{\cros} \cdot \log
{\frac{||p_v-p_{v'}||}{\ell_{\spring}^{vv'}}}\cdot
\overrightarrow{p_vp_{v'}}$$

\begin{figure}[tb]
  \centering
  \begin{minipage}[b]{.52\textwidth}
    \centering
    \subfloat[\label{fig:crossing-case-springs}{}]
    {\includegraphics[width=.64\textwidth]{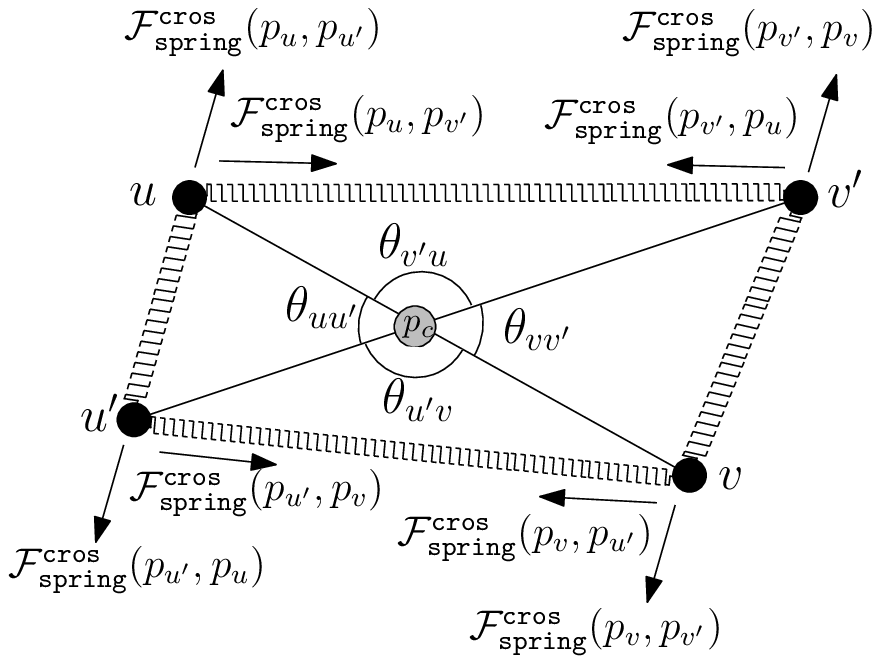}}
  \end{minipage}
  \begin{minipage}[b]{.44\textwidth}
    \centering
    \subfloat[\label{fig:crossing-case-bisection}{}]
    {\includegraphics[width=.68\textwidth]{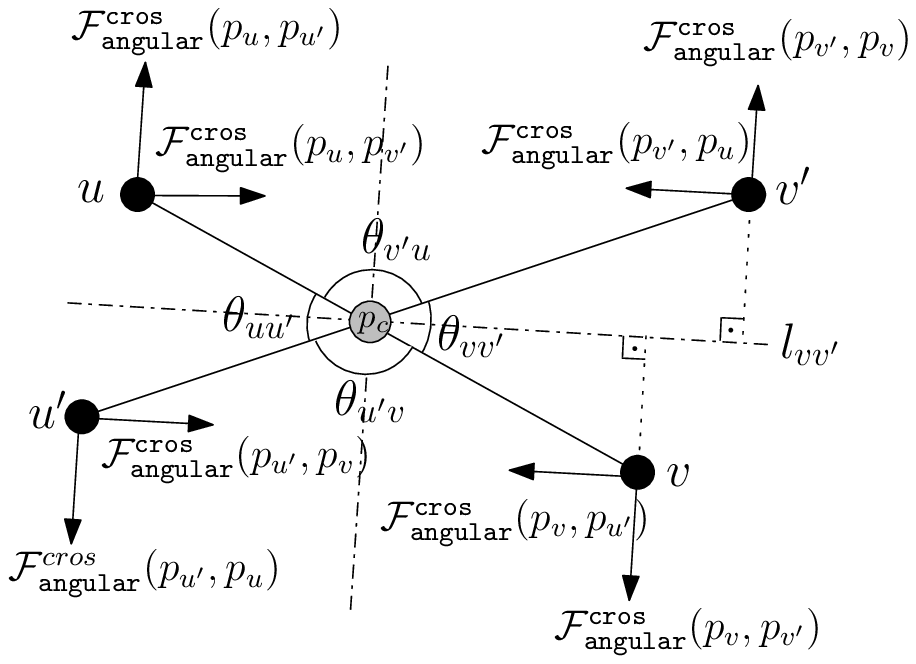}}
  \end{minipage}
  \caption{Forces applied on nodes in order to maximize the crossing resolution. (a)~Springs on nodes involved in crossing. (b)~Repelling or attractive forces based on the angles. }
  \label{fig:crossing case}
\end{figure}

The remaining forces of Fig.\ref{fig:crossing-case-springs} are
defined similarly. Note that in the formula above, the constant
$C_{\spring}^{\cros}$ is used to control the stiffness of the
springs.


Our preference for right angle crossings can be also captured using
the angles $\theta_{vv'}$ and $\theta_{v'u}$ (see
Fig.\ref{fig:crossing-case-bisection}). As in the previous case, we
restrict our description on the angle formed by the line
segments $\overline{p_cp_v}$ and $\overline{p_cp_{v'}}$. Ideally, we
would like to exert forces on the nodes $v$ and $v'$ such that:
(i)~when $\theta_{vv'} \rightarrow 0$, the magnitude of the force is
very large (in order to repel $v$ and $v'$), and (ii)~when
$\theta_{vv'} \rightarrow \frac{\pi}{2}$, the magnitude of the force
is very small. A function, say $f:\mathbb{R}\rightarrow\mathbb{R}$,
which captures this property is: $f(\theta) =
\frac{|\frac{\pi}{2}-\theta|}{\theta}$. Having specified the
magnitude of the forces,
we set the direction of the force on $v$ (due to $v'$) to be
perpendicular to the line that bisects the angle $\theta_{vv'}$
(refer to the dash-dotted line $l_{vv'}$ of
Fig.\ref{fig:crossing-case-bisection}), or equivalently parallel to
the unit length vector
$\mathtt{Perp}(\mathtt{Bsc}(\overrightarrow{p_cp_v},
\overrightarrow{p_cp_{v'}}))$. Recall that $\mathtt{Perp}$ and
$\mathtt{Bsc}$ refer to the perpendicular and bisector vectors,
respectively (see Section~\ref{sec:preliminaries}). It is clear that
if $\theta_{vv'}<\frac{\pi}{2}$, the forces on $v$ and $v'$ should
be repulsive (in order to enlarge the angle between them), otherwise
attractive. This can be captured by the $sign$ function.
We conclude with the following formula which expresses the force on
$v$ due to $v'$.
$$\mathcal{F}_{\ang}^{\cros}(p_v, p_{v'}) = C_{\ang}^{\cros} \cdot sign(\theta_{vv'}- \frac{\pi}{2}) \cdot
f(\theta_{vv'}) \cdot
\mathtt{Perp}(\mathtt{Bsc}(\overrightarrow{p_cp_v},
\overrightarrow{p_cp_{v'}}))$$

\noindent where constant $C_{\ang}^{\cros}$ controls the strength of
the force. Similarly, we define the remaining forces of
Fig.\ref{fig:crossing-case-bisection}.


Consider now a node
$u$ that is incident to $d(u)$ edges, say $e_0=(u,v_0)$,
$e_1=(u,v_1)$, $\ldots$, $e_{d(u)-1}=(u,v_{d(u)-1})$. We assume that
$e_0,e_1,\ldots,e_{d(u)-1}$ are consecutive in the counter-clockwise
order around $u$ in the drawing of the graph (see
Fig.\ref{fig:angular-spring-forces-only}). Similarly to the case of
two crossing edges, we proceed to connect the endpoints of consecutive edges around
$u$ by springs, as in Fig.\ref{fig:angular-spring-forces-only}. In
this case, the natural length of each spring, should capture our
preference for angles equal to $\frac{2\pi}{d(u)}$.
In order to achieve this, we proceed as follows: For each
$i=0,1,\ldots, d(u)-1$, we set the natural length, say
$l_{\spring}^i$, of the spring connecting $v_i$ with
$v_{(i+1)mod(d(u))}$,  to be:\\

$\ell_{\spring}^{i}=\sqrt{||e_i||^2+||e_{(i+1)mod(d(u))}||^2-2\cdot||e_i||
\cdot ||e_{(i+1)mod(d(u))}|| \cdot \cos{(2\pi/d(u))}}$\\

\noindent where $||e||$ is used to denote the length of the edge
$e\in E$ in the drawing of the graph. The quantity
$\ell_{\spring}^{i}$ corresponds to the length of the line segment
that connects $v_i$ with $v_{(i+1)mod(d(u))}$ in the optimal case
where the angle formed by $e_i$ and $e_{(i+1)mod(d(u))}$ is
$\frac{2\pi}{d(u)}$.
Therefore, the spring forces between consecutive edges should follow
the formula:\\

$\mathcal{F}_{\spring}^{\angular}(p_{v_i}, p_{v_{(i+1)mod(d(u))}};u)
= C_{\spring}^{\angular} \cdot \log
{\frac{||p_{v_i}-p_{v_{(i+1)mod(d(u))}}||}{\ell_{\spring}^{i}}}\cdot
\overrightarrow{p_{v_i}p_{v_{(i+1)mod(d(u))}}}$\\

\noindent where the quantity $C_{\spring}^{\angular}$ is a constant
which captures the stiffness of the spring.

\begin{figure}[tb]
  \centering
  \begin{minipage}[b]{.48\textwidth}
    \centering
    \subfloat[\label{fig:angular-spring-forces-only}{}]
    {\includegraphics[width=.62\textwidth]{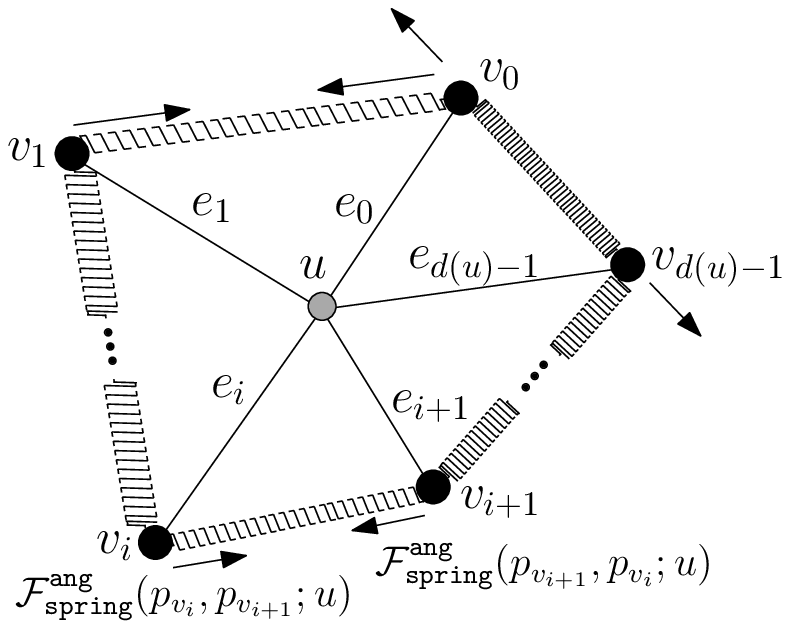}}
  \end{minipage}
  \begin{minipage}[b]{.48\textwidth}
    \centering
    \subfloat[\label{fig:angle-case-bisection}{}]
    {\includegraphics[width=.62\textwidth]{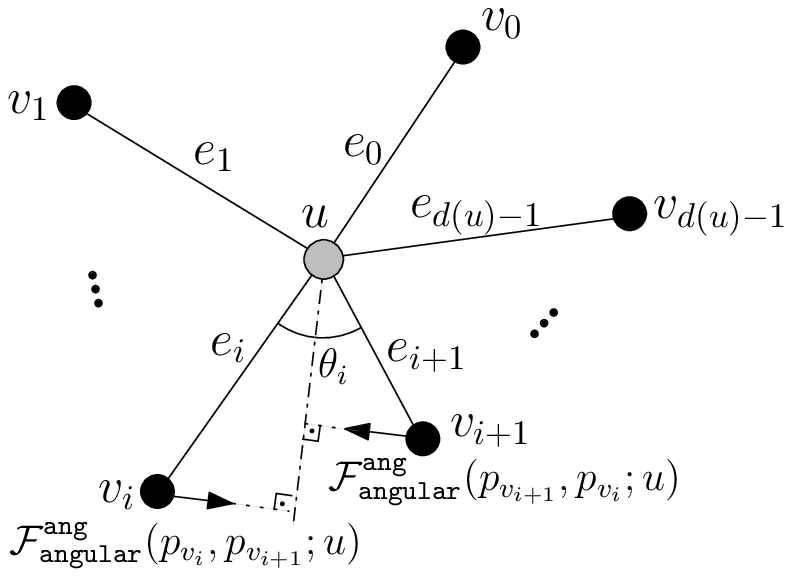}}
  \end{minipage}
  \caption{Forces applied on nodes in order to maximize the angular resolution. (a)~Springs on consecutive edges around $u$. (b)~Repelling or attractive forces based on the angles.}
  \label{fig:angular case}
\end{figure}

Let now $\theta_{i}$ be the angle formed by $e_i$ and
$e_{(i+1)mod(d(u))}$, measured in counter-clockwise direction from
$e_i$ to $e_{(i+1)mod(d(u))}$, $i=0,1,\ldots, d(u)-1$. Similarly to
the case of two crossing edges, we exert forces on $v_i$ and
$v_{(i+1)mod(d(u))}$ perpendicular to the bisector of $\theta_{i}$,
as illustrated in Fig.\ref{fig:angle-case-bisection}. However, in
this case we need a magnitude function such that: (i)~when $\theta_i
\rightarrow 0$, the magnitude of the force is very large (in order
to repel $v_i$ and $v_{(i+1)mod(d(u))}$), and (ii)~when $\theta_i
\rightarrow \frac{2\pi}{d(u)}$, the magnitude of the force is very
small. Such a function, say $g:\mathbb{R}\times
V\rightarrow\mathbb{R}$, is: $g(\theta;u) =
\frac{|\frac{2\pi}{d(u)}-\theta|}{\theta}$. Having fully specified
the forces applied on the endpoints of consecutive edges and their
directions, we are now ready to provide the exact formulas that the
forces follow:

\begin{align*}
\mathcal{F}_{\ang}^{\angular}(p_{v_i}, p_{v_{(i+1)mod(d(u))}};u)
=&~C_{\ang}^{\angular} \cdot sign(\theta_{i}- \frac{2\pi}{d(u)})
\cdot g(\theta_{i};u) \cdot \\ ~ &
\mathtt{Perp}(\mathtt{Bsc}(\overrightarrow{p_up_{v_i}},
\overrightarrow{p_{u}p_{v_{(i+1)mod(d(u))}}}))
\end{align*}

\noindent where $C_{\ang}^{\angular}$ is a constant to control the
strength of the force. In the work of Lin and Yen \cite{LY05}, the
above technique which applies large repelling forces perpendicular
to the bisector of the angle, when the angle is small, is referred
to as edge-edge repulsion. However, in their work, they use
different metric to control the magnitude function. Of course, we
could also use their metric, but we prefer the one reported above in
order to maintain a uniform approach in both crossing and angular
cases. In Section \ref{sec:experimental-results}, we provide an
experimental comparison of these techniques.

Note that by setting zero values to the constants
$C_{\spring}^{\cros}$, $C_{\ang}^{\cros}$ or
$C_{\spring}^{\angular}$, $C_{\ang}^{\angular}$, our algorithm can
be configured to maximize the angular, or the crossing resolution
only, respectively.



On each iteration, our algorithm computes three types of
forces. Computing the attractive forces of the classical
force-directed model among pairs of adjacent nodes of the graph
requires $O(E)$ time per iteration. The computation of the forces
due to the edge crossings needs $O(E^2)$ time, assuming a straight
forward algorithm that in $O(E^2)$ time reports all pairwise
crossing edges. Finally, the computation of the forces due to the
angles between consecutive edges can be done in $O(E+V
d(G)\log{d(G)})$ time per iteration, where $d(G)$ denotes the degree
of the graph, since we first sort the incident edges of each
node of the graph in cyclic order. Summarizing the above, each
iteration of our algorithm takes $O(E^2 + Vd(G)\log{d(G)})$ time.

The time complexity can be improved using standard techniques from
computational geometry \cite{PS85}. If $K$ is the number of
pairwise-crossing edges, then the $K$ intersections can be reported
in $O(K+E\log^2{E}/\log{\log{E}})$ time \cite[pp.277]{PS85}, which
leads to a total complexity $O(K+E\log^2{E}/\log{\log{E}} +
Vd(G)\log{d(G)})$ per iteration.

\eat{
\begin{algorithm2e}[t!hb]

\Input{An indirected graph $G$ and an initial placement $p=(p_v)_{v \in V}$.}

\Output{A drawing of $G$ with angular and crossing resolution  as
large as possible.}

\caption{\textsc{Force-Directed Algorithm}}

\label{alg:force-directed-algorithm}

\BlankLine

\For {($t \leftarrow 1$ \KwTo \textsc{Iterations})}
{
    \BlankLine
    \cmt{Type $1$: Attractive and repulsive forces due to springs and electric forces.}

    \ForEach{$u \in V$}
    {
       $\mathcal{F}_u(t) \leftarrow \displaystyle\sum\limits_{v: (u,v) \in E} \mathcal{F}_{\spring}(p_v,p_u) + \displaystyle\sum\limits_{v: \in V \times V} \mathcal{F}_{\ang}(p_v,p_u)$ \;
    }

    \BlankLine
    \cmt{Type $2$: Forces applied on nodes in order to maximize the crossing resolution.}

    \ForEach{(pair of intersecting edges $e=(u,v)$ and $e'=(u',v')$)}
    {
        \BlankLine
        \comment{ The relative positions of edges $e$ and $e'$ are illustrated in Fig.\ref{fig:crossing case}.}\;
        \BlankLine
        $\mathcal{F}_{v}(t) += \mathcal{F}_{\spring}^{\cros}(p_v, p_{v'}) + \mathcal{F}_{\spring}^{\cros}(p_v, p_{u'}) + \mathcal{F}_{\ang}^{\cros}(p_v, p_{v'}) + \mathcal{F}_{\ang}^{\cros}(p_v, p_{u'})$\;
        \BlankLine
        $\mathcal{F}_{v'}(t) += \mathcal{F}_{\spring}^{\cros}(p_{v'}, p_u) + \mathcal{F}_{\spring}^{\cros}(p_{v'}, p_v) + \mathcal{F}_{\ang}^{\cros}(p_{v'}, p_{u}) + \mathcal{F}_{\ang}^{\cros}(p_{v'}, p_{v})$\;
        \BlankLine
        $\mathcal{F}_{u}(t) += \mathcal{F}_{\spring}^{\cros}(p_u, p_{v'}) + \mathcal{F}_{\spring}^{\cros}(p_u, p_{u'}) + \mathcal{F}_{\ang}^{\cros}(p_u, p_{v'}) + \mathcal{F}_{\ang}^{\cros}(p_u, p_{u'})$ \;
        \BlankLine
        $\mathcal{F}_{u'}(t) += \mathcal{F}_{\spring}^{\cros}(p_{u'}, p_{v}) + \mathcal{F}_{\spring}^{\cros}(p_{u'}, p_{u}) + \mathcal{F}_{\ang}^{\cros}(p_{u'}, p_{v}) + \mathcal{F}_{\ang}^{\cros}(p_{u'}, p_{v})$\;
    }

    \BlankLine
    \cmt{Type $3$: Forces applied on nodes in order to maximize the angular resolution.}

    \ForEach{($u \in V$ with $d(u)>1$)}
    {
        \BlankLine
        $e_0, \ldots, e_{d(u)-1}$ $\leftarrow$ incident edges of $u$ in counter-clockwise order (see Fig.\ref{fig:angular case})\;
        \BlankLine
        Let $e_i=(u,v_i), i=0, \ldots d(u)-1$\;
        \BlankLine
        \For {($i \leftarrow 0$ \KwTo $d(u)-1$)}
        {
            $\mathcal{F}_{u}(t) += \mathcal{F}_{\spring}^{\angular}(p_{v_i}, p_{v_{(i+1)mod(d(u))}};u) + \mathcal{F}_{\ang}^{\angular}(p_{v_i}, p_{v_{(i+1)mod(d(u))}};u)$\;
        }
    }

    \BlankLine
    \ForEach{$u \in V$}
    {
        \BlankLine
        $p_u \leftarrow p_u + \delta \cdot \mathcal{F}_u(t)$ \;
    }
}

\end{algorithm2e}

}

\subsection{Experimental Results}
\label{sec:experimental-results}

In this section, we present the results of the experimental
evaluation of our algorithm. Apart from our algorithm, we have
implemented the force directed algorithms of Eades~\cite{Ea84} and
Lin and Yen \cite{LY05}. The implementations are in Java using the
yFiles library (http://www.yworks.com). The experiment was performed
on a Linux machine with 2.00 GHz CPU and 2GB RAM using the Rome
graphs (a collection of around 11.500 graphs) obtained from
graphdrawing.org. Fig.\ref{fig:grafo1012999}, illustrates a drawing
of a Rome graph with $99$ nodes and $135$ edges produced by our
force directed algorithm.

\begin{figure}[htb]
  \centering
  \includegraphics[width=.68\textwidth]{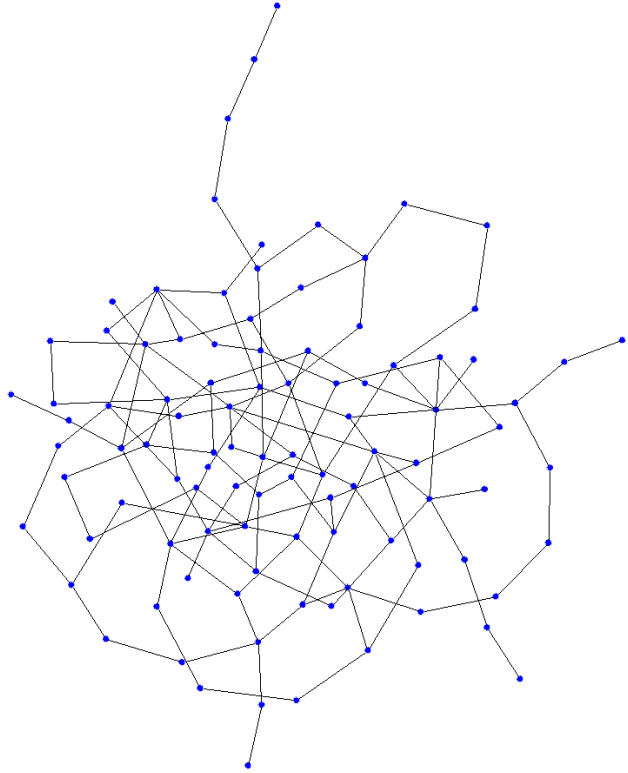}
  \caption{A drawing of Rome graph grafo10129.99 consisting of $99$ nodes and $135$ edges with angular resolution $20.15^o$ and crossing resolution $26.12^o$.}
  \label{fig:grafo1012999}
\end{figure}

The experiment was performed as follows. First, each Rome graph was
laid out using the SmartOrganic layouter of yFiles. This layout was
the input layout for all three algorithms, in order to speed up the
experiment and overcome problems associated with local minimal traps
especially in large graphs. If both the angular and the crossing
resolution between two consecutive iterations of each algorithm were
not improved more that 0.001 degrees, we assumed that the algorithm
has converged and we did not proceed any more. The maximum number of
iterations that an algorithm could perform in order to converge was
100.000. We note that the termination condition is quite strict and
demands a large number of iterations.
Our algorithm is evaluated as (a)~Crossing-Only, (b)~Angular-Only
and (c)~Mixed. The results are illustrated in
Fig.\ref{fig:experiment-results} and should be viewed in color.

\emph{The Angular Resolution Maximization Problem:} Refer to
Fig.\ref{fig:angular-resolution}. Our experimental analysis shows
that our Angular-Only algorithm achieves, on average, better angular
resolution. The angular resolution of our Mixed algorithm is almost
equal, on average, to the one of Lin and Yen. Note that the
algorithm of Lin and Yen, in contrast to ours, does not modify the
embedding of the initial layout~\cite{LY05} (i.e., it needs a
close-to-final starting layout and improve on it). This explains why
our Mixed algorithm achieves almost the same performance, in terms
of angular resolution, as the one of Lin and Yen. In $59.76$\% of
the graphs our Mixed algorithm yields a better solution compared to
Lin-Yen's algorithm with an average improvement of $6.94^o$.

\emph{The Crossing Resolution Maximization Problem:} In
Fig.\ref{fig:crossing-resolution} the data were filtered to depict
only the results of non-planar drawings produced by the algorithms
and avoid infinity values in the case of planar ones. It is clear
that our Crossing-Only algorithm results in drawings with high
crossing resolution. Our Crossing-Only algorithm performs better on
large graphs compared to the algorithm of yFiles. The average
improvement implied by our Crossing-Only algorithm is $13.63^o$
w.r.t. the yFiles algorithm.

\emph{The Total Resolution Maximization Problem:} Refer to
Fig.\ref{fig:total-resolution}. This is the most important result of
our experimental analysis. It indicates that our Mixed algorithm
applied on graphs with more than $50$ nodes constructs drawings with
total resolution of $20$ degrees, on average. Note that this is an
achievement on average and under the particular termination
condition discussed above, i.e., there is no guarantee, that it
would be maintained indefinitely. An example is given in
Fig.\ref{fig:grafo1012999}.





Finally, Fig.\ref{fig:algorithm-time} summarizes the running time
performance of the algorithms. Our algorithm needs, on average,
$7340$ milliseconds and $1298$ iterations to converge, whereas the
one of Lin and Yen $8346$ and $1952$, respectively. Note that the
time complexity of Lin-Yen's algorithm is better than ours. However,
the termination condition takes into account the crossing
improvement and therefore the algorithm of Lin and Yen needs more
iterations to converge, which explains this contradiction.

\begin{figure}[tb]
  \centering

  \begin{minipage}[b]{.49\textwidth}
    \centering
    \subfloat[\label{fig:angular-resolution}{Angular resolution results}]
    {\includegraphics[width=.75\textwidth]{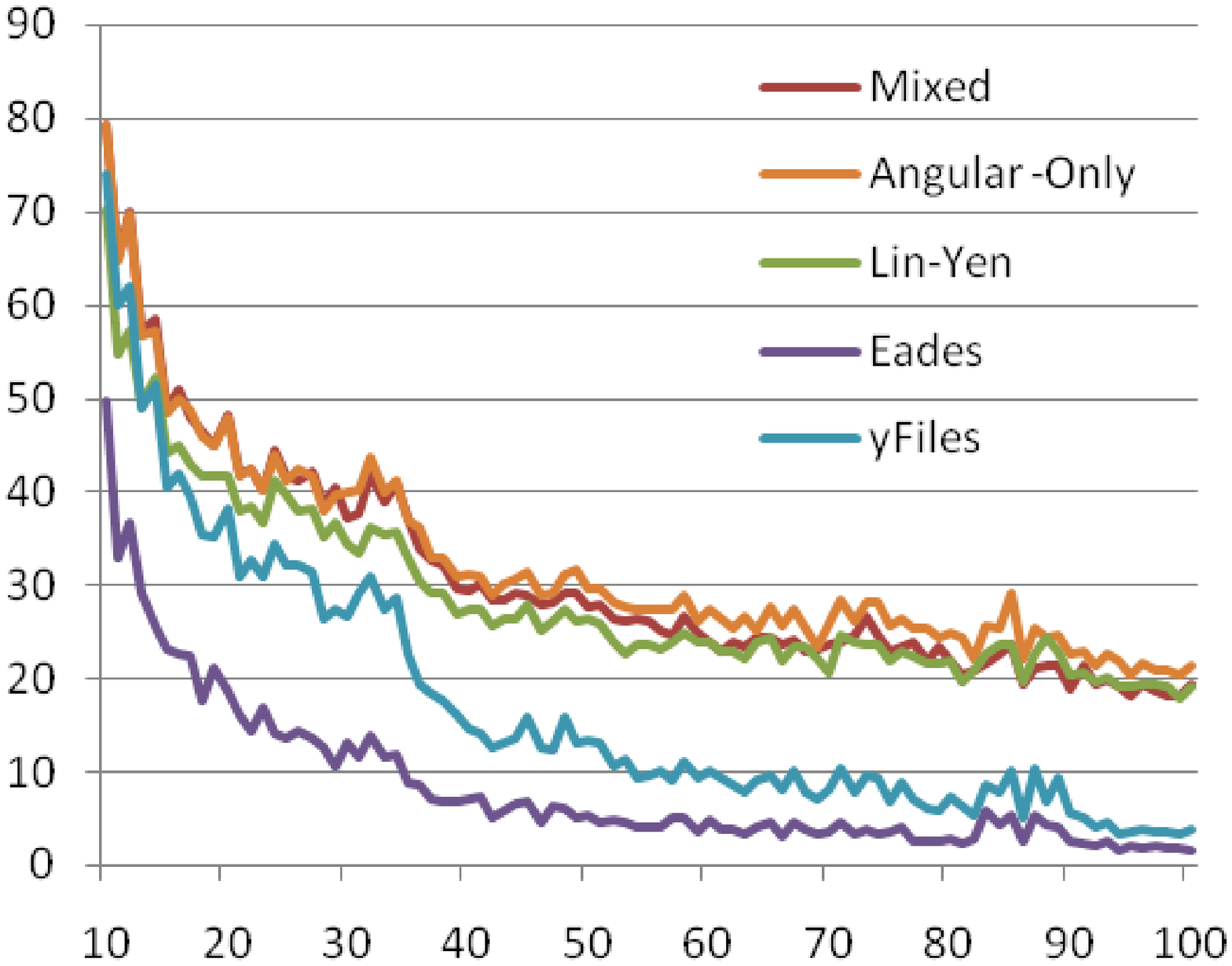}}
  \end{minipage}
  \begin{minipage}[b]{.49\textwidth}
    \centering
    \subfloat[\label{fig:crossing-resolution}{Crossing resolution results}]
    {\includegraphics[width=.75\textwidth]{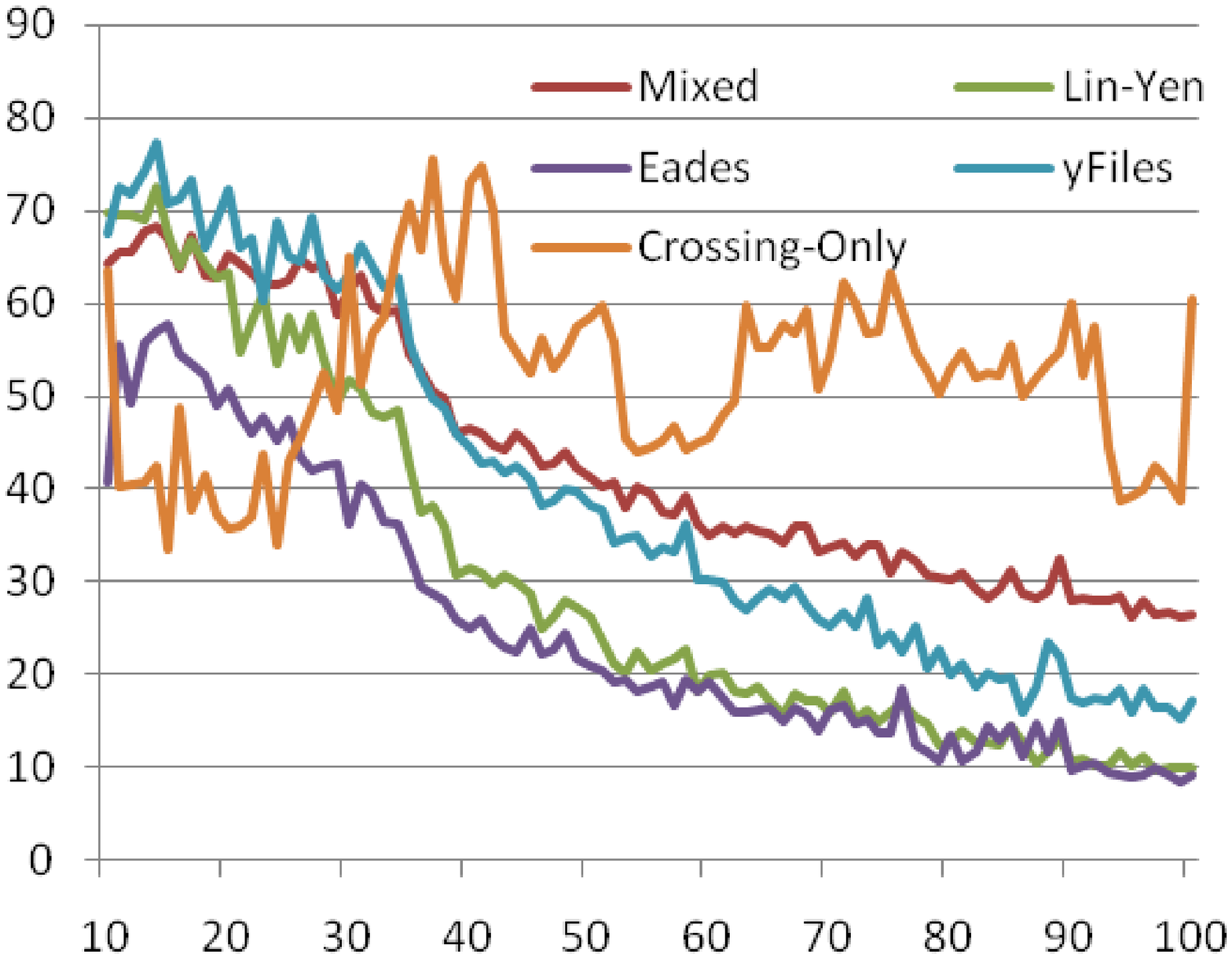}}
  \end{minipage}

  \begin{minipage}[b]{.49\textwidth}
    \centering
    \subfloat[\label{fig:total-resolution}{Total resolution results}]
    {\includegraphics[width=.75\textwidth]{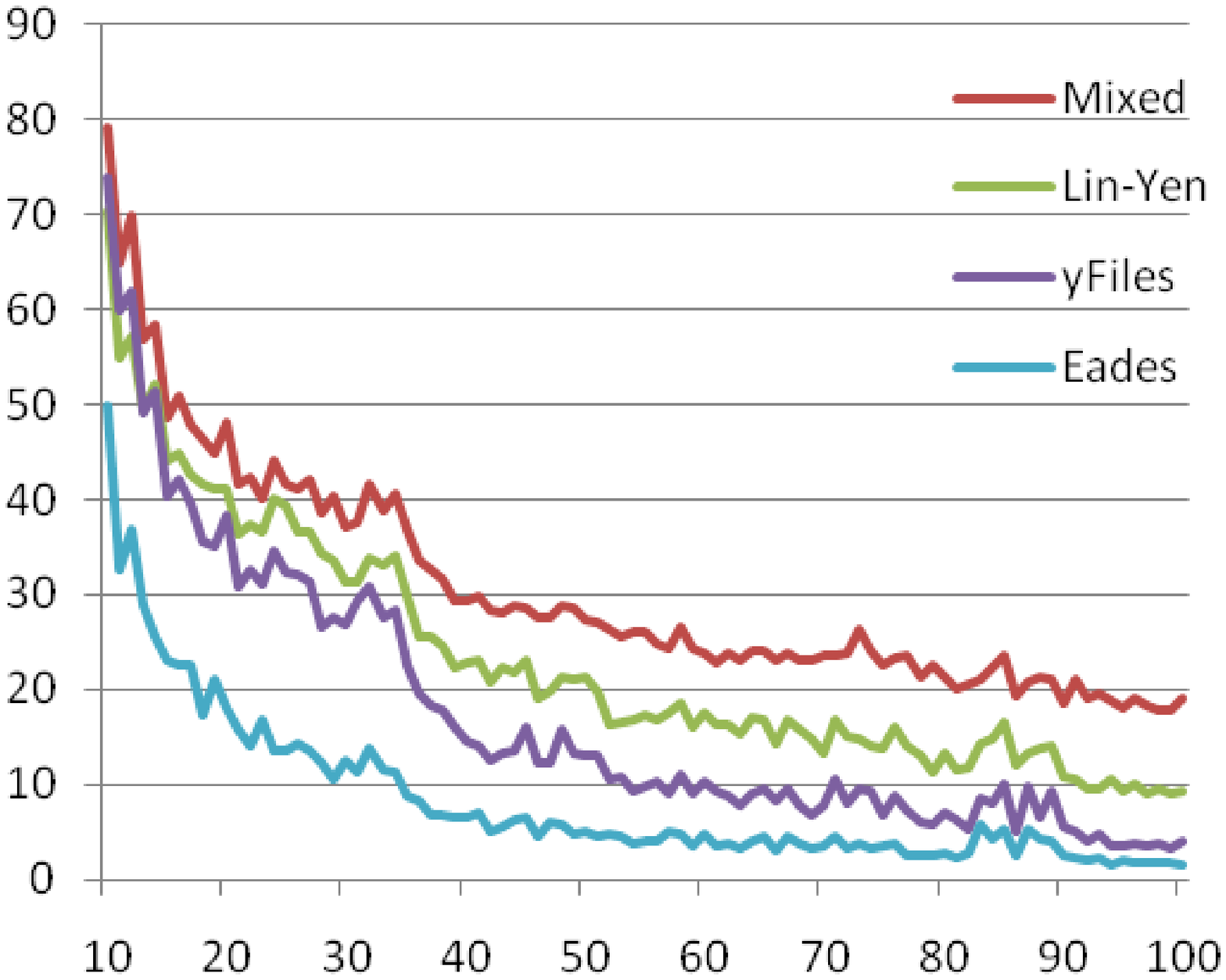}}
  \end{minipage}
  \begin{minipage}[b]{.49\textwidth}
    \centering
    \subfloat[\label{fig:algorithm-time}{Running time results}]
    {\includegraphics[width=.75\textwidth]{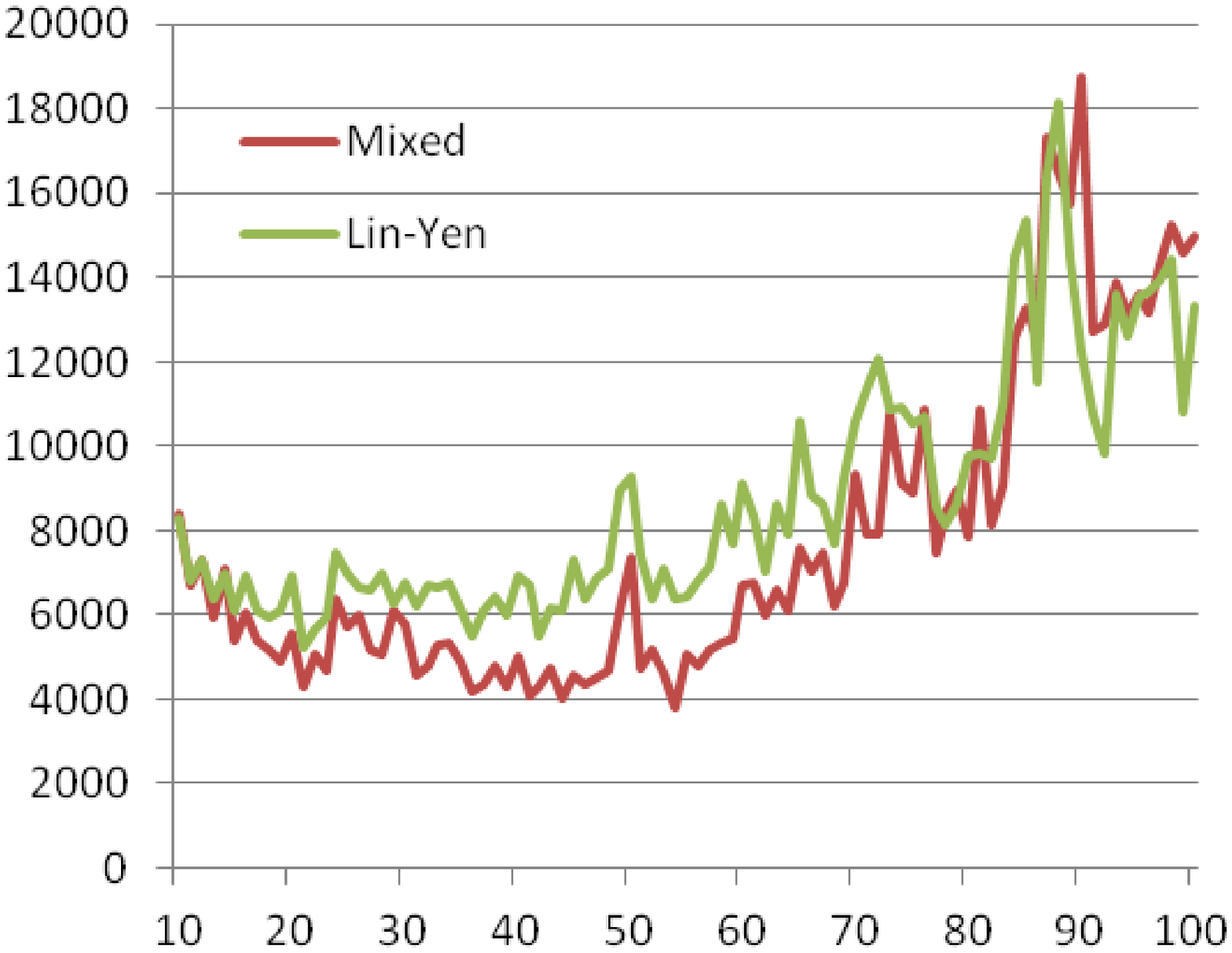}}
  \end{minipage}

  \caption{A visual presentation of our experimental results. The $X$-axis indicates the number of the nodes of the graph.
  In Fig.(a)-(c) the $Y$-axis corresponds to the resolution measured in degrees, whereas
  in Fig.(d) to the running time measured in milliseconds.}
  \label{fig:experiment-results}
\end{figure}

\section{Conclusions}
\label{sec:conclusions}
In this paper, we introduced and studied the total resolution
maximization problem. Of course, our work leaves several open
problems. It would be interesting to try to identify other classes
of graphs that admit optimal drawings. Even the case of planar
graphs is of interest, as by allowing some edges to cross (say at
large angles), we may improve the angular resolution and therefore
the total resolution.

\newpage
\newpage

\bibliographystyle{abbrv}

\end{document}